\documentclass[12pt]{article}

\usepackage[margin=1in]{geometry}
\usepackage[basic]{complexity}
\usepackage{CJKutf8, amsmath, amssymb, amsthm, authblk, hyperref}

\newtheorem{lemma}{Lemma}
\newtheorem{theorem}{Theorem}

\theoremstyle{definition}
\newtheorem{assumption}{Assumption}

\DeclareMathOperator{\tr}{tr}

\usepackage[style=trad-unsrt, citestyle=numeric-comp, giveninits=true, doi=false, url=false, eprint=false, isbn=false]{biblatex}
\begin{document}

\title{Finite-size scaling analysis of eigenstate thermalization}

\begin{CJK}{UTF8}{gbsn}

\author{Yichen Huang (黄溢辰)\thanks{yichuang@mit.edu}}
\affil{Center for Theoretical Physics, Massachusetts Institute of Technology, Cambridge, Massachusetts 02139, USA}

\maketitle

\end{CJK}

\begin{abstract}

We study the fluctuations of eigenstate expectation values in a microcanonical ensemble. Assuming the eigenstate thermalization hypothesis, an analytical formula for the finite-size scaling of the fluctuations is derived. The same problem was studied by Beugeling et al. [Phys. Rev. E 89, 042112 (2014)]. We compare our results with theirs.

\end{abstract}

\section{Introduction}

As an explanation for the emergence of statistical mechanics from the unitary evolution of isolated quantum many-body systems, the eigenstate thermalization hypothesis (ETH) \cite{Deu91, Sre94, RDO08, DKPR16, DLL18, Deu18} is of high current interest. A large body of literature is available on testing the ETH \cite{KIH14, SV17, HLZ18, GG18, BLGR20, RDSG20, SHU21} and exploring its implications \cite{KLW15, VR17, LCB18, Hua19NPB, HG19, LG19, BCSB19, WGRE19, ARG20, Hua21NPB, Hua21ISIT, HMK22, Hua22TIT, HH19, Hua19aA, CB21} in various systems.

The ETH states that eigenstates that are close in energy have similar local expectation values, i.e., the fluctuations of eigenstate expectation values (EEV) in a microcanonical ensemble vanish in the thermodynamic limit. What is the asymptotic behavior of EEV fluctuations as the system size diverges? The answer to this question depends on how EEV fluctuations are defined. Assuming the ETH, we rigorously derive an analytical formula for the finite-size scaling of some definitions (including the one in Ref. \cite{BMH14}) of EEV fluctuations.

The rest of this paper is organized as follows. Section \ref{s:def} sets the stage and introduces some definitions of EEV fluctuations. Section \ref{s:res} presents the main results. Section \ref{s:dis} compares our results with those of Beugeling et al. \cite{BMH14} for the same problem. The main text of this paper should be easy to read, for most of the technical details are deferred to Appendix \ref{s:app}.

\section{Definitions} \label{s:def}

Throughout this paper, standard asymptotic notations are used extensively. Let $f,g:\mathbb R^+\to\mathbb R^+$ be two functions. One writes $f(x)=O(g(x))$ if and only if there exist constants $M,x_0>0$ such that $f(x)\le Mg(x)$ for all $x>x_0$; $f(x)=\Omega(g(x))$ if and only if there exist constants $M,x_0>0$ such that $f(x)\ge Mg(x)$ for all $x>x_0$; $f(x)=\Theta(g(x))$ if and only if there exist constants $M_1,M_2,x_0>0$ such that $M_1g(x)\le f(x)\le M_2g(x)$ for all $x>x_0$.

Consider a system of $N$ spins on a hypercubic lattice in $D=\Theta(1)$ spatial dimensions, where each lattice site has a spin. The dimension of the Hilbert space is $d=d_\textnormal{loc}^N$, where $d_\textnormal{loc}=\Theta(1)$ is the local dimension of each spin. The system is governed by a (not necessarily translation-invariant) local Hamiltonian
\begin{equation} \label{eq:H}
H=\sum_iH_i.    
\end{equation}
The sum is over $\Theta(N)$ lattice sites. Each term $H_i$ has operator norm $\|H_i\|=\Theta(1)$ and is supported in a small neighborhood of site $i$. Assume without loss of generality that $\tr H_i=0$ (traceless) so that the mean energy of $H$ is $\tr H/d=0$.

Let $\{|j\rangle\}_{j=1}^d$ be a complete set of eigenstates of $H$ with corresponding energies $\{E_j\}$. Let $J:=\{j:-N\delta_1\le E_j\le N\delta_2\}$ be a microcanonical ensemble in the middle of the energy spectrum, where $\delta_1,\delta_2=\Theta(1)$ are arbitrary positive constants. Let $A$ be a traceless local operator with $\|A\|=1$ and $A_{jj}:=\langle j|A|j\rangle$ be the EEV so that $\sum_{j=1}^dA_{jj}=\tr A=0$.

In this paper, we consider three definitions of EEV fluctuations. The first
\begin{equation}
\Delta A:=\frac{1}{d}\sum_{j=1}^d|A_{jj}|^2
\end{equation}
is simply the variance of $A_{jj}$ in all eigenstates. The second
\begin{equation}
\Delta A_J:=\frac{1}{|J|}\sum_{j\in J}\left|A_{jj}-\frac{1}{|J|}\sum_{k\in J}A_{kk}\right|^2=\frac{1}{|J|}\sum_{j\in J}|A_{jj}|^2-\frac{1}{|J|^2}\left|\sum_{j\in J}A_{jj}\right|^2
\end{equation}
is the variance of $A_{jj}$ in $J$. The third definition \cite{BMH14} is slightly more complicated. Let $K_j=\{k:|E_j-E_k|\le N\delta_3\}$, where $\delta_3=\Theta(1)$ is an arbitrary positive constant. Let
\begin{equation}
\Delta'A_J:=\frac{1}{|J|}\sum_{j\in J}\left|A_{jj}-\frac{1}{|K_j|}\sum_{k\in K_j}A_{kk}\right|^2.
\end{equation}

\section{Results} \label{s:res}

In the thermodynamic limit $N\to+\infty$, the fluctuations $\Delta A_J,\Delta'A_J$ depend weakly on the hyperparameters $\delta_1,\delta_2,\delta_3$, and are approximately equal to $\Delta A$ up to exponentially small additive errors.

\begin{lemma} \label{l:1}
For any traceless local operator $A$ with $\|A\|=1$,
\begin{gather}
|\Delta A_J-\Delta A|=e^{-\Omega(N)}, \label{eq:def1} \\
|\Delta'A_J-\Delta A|=e^{-\Omega(N)}. \label{eq:def2} 
\end{gather}
\end{lemma}

\begin{proof}
See Appendix \ref{ss:pl1}.
\end{proof}

It suffices to assume the ETH for eigenstates in the middle of the energy spectrum.

\begin{assumption} [eigenstate thermalization hypothesis in the middle of the spectrum \cite{HBZ19, Hua20}] \label{ethasmp}
Let $\epsilon$ be an arbitrarily small positive constant. For any traceless local operator $A$ with $\|A\|=1$, there is a function $f_A:[-\epsilon,\epsilon]\to\{z\in\mathbb C:|z|\le1\}$ such that
\begin{equation} \label{asmpeq}
|A_{jj}-f_A(E_j/N)|\le1/\poly(N)
\end{equation}
for all $j$ with $|E_j|\le N\epsilon$, where $\poly(N)$ denotes a polynomial of sufficiently high degree in $N$. We assume that $f_A(x)$ is smooth in the sense of having a Taylor expansion to some low order around $x=0$.
\end{assumption}

In quantum chaotic systems, it was proposed analytically \cite{Sre99} and supported by numerical simulations \cite{KIH14} that the right-hand side of (\ref{asmpeq}) can be improved to $e^{-\Omega(N)}$. For our purposes, however, a (much weaker) inverse polynomial upper bound suffices.

\begin{lemma} \label{l:thm}
For a traceless local operator $A$ with $\|A\|=1$, Assumption \ref{ethasmp} implies that
\begin{equation} \label{eq:main}
\Delta A=\frac{|\tr(HA)|^2}{d\tr(H^2)}+O(1/N^2).
\end{equation}
\end{lemma}

For generic $A$, $\tr(HA)$ is non-zero, and hence the first term on the right-hand side of Eq. (\ref{eq:main}) is $\Theta(1/N)$ (see Eq. (\ref{eq:m1})).

Equation (\ref{eq:main}) was derived in Refs. \cite{HBZ19, MV20}, neither of which has rigorously bounded the approximation error in the derivation. In Appendix \ref{ss:pl2}, we follow Ref. \cite{HBZ19} and present a complete proof of Lemma \ref{l:thm} with rigorous error analysis.

To understand Lemmas \ref{l:1}, \ref{l:thm} quickly without going into technical details, it is instructive to consider the special case where the Hamiltonian (\ref{eq:H}) is translation invariant and where $A=H_i$ is a term in the Hamiltonian. In this case, Lemma \ref{l:thm} is trivially true since $(H_i)_{jj}=E_j/N$. Lemma \ref{l:1} follows directly from the fact that $E_j$'s approach a normal distribution in the thermodynamic limit $N\to+\infty$ \cite{KLW15, BC15}.

Combining Lemmas \ref{l:1} and \ref{l:thm}, we obtain
\begin{theorem} \label{thm:cor}
For a traceless local operator $A$ with $\|A\|=1$, Assumption \ref{ethasmp} implies that
\begin{equation} \label{eq:cor}
\Delta A_J=\frac{|\tr(HA)|^2}{d\tr(H^2)}+O(1/N^2),\quad\Delta'A_J=\frac{|\tr(HA)|^2}{d\tr(H^2)}+O(1/N^2).
\end{equation}
\end{theorem}

\section{Discussion} \label{s:dis}

The finite-size scaling of the EEV fluctuation $\Delta'A_J$ was studied by Beugeling et al. \cite{BMH14}. For $\delta_3=0.025$ and $J$ being the middle $20\%$ of the spectral range, they presented numerical evidence that $\Delta'A_J$ scales as $d^{-1/2}=e^{-\Theta(N)}$ in generic non-integrable systems, where $d$ is the dimension of the Hilbert space. Furthermore, a heuristic ``typicality argument'' was given to explain the $d^{-1/2}$ behavior. This result is different from ours. Recall that Eq. (\ref{eq:cor}) states that for generic $A$, $\Delta'A_J$ scales as $\Theta(1/N)$.

In fact, the typicality argument has already been proved to be problematic in a different but related setting \cite{HU18}. Using the equivalence of $\Delta A_J$ and $\Delta' A_J$ (Lemma \ref{l:1}), the proof can be extended to $\Delta' A_J$. This requires some work, and we do not present the extension here.

For a constant $\delta_3$\footnote{I have confirmed by personal communication with the authors of Ref. \cite{SHU21} that $\delta_3$ is set as a constant in their numerical study.} and $J$ being the middle $10\%$ of the spectral range, Sugimoto et al. \cite{SHU21} found numerically that
\begin{equation} \label{eq:Japan}
\Delta''A_J=e^{-\Omega(N)},\quad\Delta''A_J:=\max_{j\in J}\left|A_{jj}-\frac{1}{|K_j|}\sum_{k\in K_j}A_{kk}\right|^2
\end{equation}
in generic translation-invariant and generic disordered spin chains. Since $\Delta''A_J\ge\Delta'A_J$, Eq. (\ref{eq:Japan}) is inconsistent with Eq. (\ref{eq:cor}).

We suspect that the discrepancy between Eq. (\ref{eq:cor}) and the numerical results of Refs. \cite{BMH14, SHU21} is due to finite-size effects in the simulations.

\section*{Declaration of competing interest}

The author declares that he has no known competing financial interests or personal relationships that could have appeared to influence the work reported in this paper.

\section*{Acknowledgments}

The author would like to thank Fernando G.S.L. Brand\~ao, Xie Chen, and Yong-Liang Zhang for collaboration on a related project \cite{HBZ19}. This work was supported by NSF grant PHY-1818914.

\appendix 

\section{Proofs} \label{s:app}

\begin{lemma} [moments \cite{HBZ19}] \label{l:m}
For any integer $m\ge0$,
\begin{gather}
\frac{1}{d}\sum_jE_j^{2m}=\frac{1}{d}\tr(H^{2m})=\Theta(N^m), \label{eq:m1}\\
\frac{1}{d}\left|\sum_jE_j^3\right|=\frac{1}{d}|\tr(H^3)|=O(N). \label{eq:m2}
\end{gather}
\end{lemma}

\begin{proof}
Expanding $H$ in the generalized Pauli basis, we count the number of terms that do not vanish upon taking the trace in the expansion of $H^{2m}$. There are $\Theta(N^m)$ such terms, the trace of each of which is $\Theta(d)$. Therefore, we obtain Eq. (\ref{eq:m1}). Equation (\ref{eq:m2}) can be proved in the same way.
\end{proof}

Almost all eigenstates have vanishing energy density:
\begin{lemma} [concentration of eigenvalues \cite{Ans16}] \label{Mar}
For any $\epsilon>0$,
\begin{equation}
|\{j:|E_j|\ge N\epsilon\}|/d=e^{-\Omega(N\epsilon^2)}.
\end{equation}
\end{lemma}
This lemma allows us to upper bound the total contribution of all eigenstates away from the middle of the spectrum. Let $C=O(1)$ be a sufficiently large constant such that
\begin{equation} \label{tail}
\frac{1}{d}\sum_{j:|E_j|\ge\Lambda}1\le\frac{1}{d}\sum_{j:|E_j|\ge\Lambda}|E_j|\le q,\quad\Lambda:=C\sqrt{N\log N},\quad q:=1/\poly(N),
\end{equation}
where $\poly(N)$ denotes a polynomial of sufficiently high degree in $N$.

Lemma \ref{Mar}, Eq. (\ref{eq:m1}), and (\ref{tail}) are related to the fact that $E_j$'s approach a normal distribution in the thermodynamic limit $N\to+\infty$ \cite{KLW15, BC15}. Indeed, $|E_j|=\Theta(\sqrt N)$ for almost all $j$.

\subsection{Proof of Lemma \ref{l:1}} \label{ss:pl1}

\begin{proof} [Proof of Eq. (\ref{eq:def1})]
By definition,
\begin{multline}
\Delta A_J-\Delta A=\frac{1}{|J|}\sum_{j\in J}|A_{jj}|^2-\frac{1}{|J|^2}\left|\sum_{j\in J}A_{jj}\right|^2-\frac{1}{d}\sum_{j}|A_{jj}|^2\\
=\left(\frac{1}{|J|}-\frac{1}{d}\right)\sum_{j\in J}|A_{jj}|^2-\frac{1}{|J|^2}\left|\sum_{j\not\in J}A_{jj}\right|^2-\frac{1}{d}\sum_{j\not\in J}|A_{jj}|^2.
\end{multline}
Since $\delta_1,\delta_2=\Theta(1)$ are positive constants, Lemma \ref{Mar} implies that $1-|J|/d=e^{-\Omega(N)}$. Therefore,
\begin{equation}
|\Delta A_J-\Delta A|\le(1/|J|-1/d)|J|+(d-|J|)^2/|J|^2+(d-|J|)/d=e^{-\Omega(N)}.
\end{equation}
\end{proof}

\begin{proof} [Proof of Eq. (\ref{eq:def2})]
Let $\delta':=\min\{\delta_1,\delta_2,\delta_3/2\}=\Theta(1)$ and $J':=\{j:|E_j|\le N\delta'\}\subseteq J$. By definition,
\begin{align}
\Delta'A_J-\Delta A&=\frac{1}{|J|}\sum_{j\in J}\left|A_{jj}-\frac{1}{|K_j|}\sum_{k\in K_j}A_{kk}\right|^2-\frac{1}{d}\sum_{j}|A_{jj}|^2
\nonumber\\
&=\frac{1}{|J|}\sum_{j\in J'}\left(\frac{A_{jj}^*}{|K_j|}\sum_{k\not\in K_j}A_{kk}+\frac{A_{jj}}{|K_j|}\sum_{k\not\in K_j}A_{kk}^*+\frac{1}{|K_j|^2}\left|\sum_{k\not\in K_j}A_{kk}\right|^2\right)\nonumber\\
&+\left(\frac{1}{|J|}-\frac{1}{d}\right)\sum_{j\in J'}|A_{jj}|^2+\frac{1}{|J|}\sum_{j\in J\setminus J'}\left|A_{jj}-\frac{1}{|K_j|}\sum_{k\in K_j}A_{kk}\right|^2-\frac{1}{d}\sum_{j\not\in J'}|A_{jj}|^2.
\end{align}
For any $j\in J'$,
\begin{equation}
K_j=\{k:|E_j-E_k|\le N\delta_3\}\supseteq\{k:|E_k|\le N\delta_3/2\}.
\end{equation}
For such $j$, Lemma \ref{Mar} implies that $1-|K_j|/d=e^{-\Omega(N)}$. Therefore,
\begin{equation}
|\Delta'A_J-\Delta A|
\le\frac{3}{|J|}\sum_{j\in J'}\frac{d-|K_j|}{|K_j|}+\left(\frac{1}{|J|}-\frac{1}{d}\right)|J'|+\frac{4(|J|-|J'|)}{|J|}+\frac{d-|J'|}{d}\\
=e^{-\Omega(N)}.
\end{equation}
\end{proof}

\subsection{Proof of Lemma \ref{l:thm}} \label{ss:pl2}

\begin{lemma} [\cite{Hua20}] \label{l:2}
For a traceless local operator $A$ with $\|A\|=1$, Assumption \ref{ethasmp} implies that
\begin{gather}
|f_A(0)|=O(1/N),\label{zero}\\
|f'_A(0)-N\tr(HA)/\tr(H^2)|=O(1/N).\label{deriv}
\end{gather}
\end{lemma}

We include the proof of this lemma for completeness. For notational simplicity, let $x\overset{\delta}=y$ denote $|x-y|\le\delta$.

\begin{proof} [Proof of Eq. (\ref{zero})]
\begin{multline} \label{zeroc}
0=\frac{1}{d}\tr A=\frac{1}{d}\sum_jA_{jj}\overset{O(q)}=\frac{1}{d}\sum_{j:|E_j|<\Lambda}A_{jj}\overset{1/\poly(N)}=\frac{1}{d}\sum_{j:|E_j|<\Lambda}f_A(E_j/N)\\
\approx\frac{1}{d}\sum_{j:|E_j|<\Lambda}\left(f_A(0)+\frac{f'_A(0)E_j}{N}\right)\overset{O(q)}=\frac{1}{d}\sum_j\left(f_A(0)+\frac{f'_A(0)E_j}{N}\right)=f_A(0),
\end{multline}
where we used (\ref{tail}), the ETH (\ref{asmpeq}), and the Taylor expansion
\begin{equation} \label{taylor}
f_A(E_j/N)=f_A(0)+f'_A(0)E_j/N+f''_A(0)E_j^2/(2N^2)+O(|E_j|^3/N^3)
\end{equation}
in the steps marked with ``$O(q)$,'' ``$1/\poly(N)$,'' and ``$\approx$,'' respectively. The approximation error in the ``$\approx$'' step is
\begin{equation} \label{zeroe}
\frac{1}{d}\sum_{j:|E_j|<\Lambda}O(E_j^2/N^2)\le\frac{1}{d}\sum_jO(E_j^2/N^2)=O(1/N),
\end{equation}
where we used Eq. (\ref{eq:m1}) with $m=1$. We obtain Eq. (\ref{zero}) by combining (\ref{zeroc}), (\ref{zeroe}).
\end{proof}

\begin{proof} [Proof of Eq. (\ref{deriv})]
\begin{align} \label{derivc}
&\frac{1}{d}\tr(HA)=\frac{1}{d}\sum_jE_jA_{jj}\overset{O(q)}=\frac{1}{d}\sum_{j:|E_j|<\Lambda}E_jA_{jj}\overset{1/\poly(N)}=\frac{1}{d}\sum_{j:|E_j|<\Lambda}E_jf_A(E_j/N)\nonumber\\
&\approx\frac{1}{d}\sum_{j:|E_j|<\Lambda}\left(E_jf_A(0)+\frac{f'_A(0)E_j^2}{N}+\frac{f''_A(0)E_j^3}{2N^2}\right)\nonumber\\
&\overset{O(q)}=\frac{1}{d}\sum_j\left(E_jf_A(0)+\frac{f'_A(0)E_j^2}{N}+\frac{f''_A(0)E_j^3}{2N^2}\right)\overset{O(1/N)}=\frac{f_A'(0)\tr(H^2)}{Nd},
\end{align}
where we used (\ref{tail}), (\ref{asmpeq}), the Taylor expansion (\ref{taylor}), and Eq. (\ref{eq:m2}) in the steps marked with ``$O(q)$,'' ``$1/\poly(N)$,'' ``$\approx$,'' and ``$O(1/N)$,'' respectively. The approximation error in the ``$\approx$'' step is
\begin{equation} \label{derive}
\frac{1}{d}\sum_{j:|E_j|<\Lambda}O(E_j^4/N^3)\le\frac{1}{d}\sum_jO(E_j^4/N^3)=O(1/N),
\end{equation}
where we used Eq. (\ref{eq:m1}) with $m=2$. We obtain Eq. (\ref{deriv}) by combining (\ref{derivc}), (\ref{derive}).
\end{proof}

We are ready to prove Lemma \ref{l:thm}:
\begin{align} \label{eq:mc}
&\Delta_A=\frac{1}{d}\sum_j|A_{jj}|^2\overset{O(q)}=\frac{1}{d}\sum_{j:|E_j|<\Lambda}|A_{jj}|^2\overset{1/\poly(N)}=\frac{1}{d}\sum_{j:|E_j|<\Lambda}|f_A(E_j/N)|^2\nonumber\\
&\approx\frac{1}{d}\sum_{j:|E_j|<\Lambda}\left|f_A(0)+\frac{f'_A(0)E_j}{N}\right|^2\overset{O(q)}=\frac{1}{d}\sum_j\left|f_A(0)+\frac{f'_A(0)E_j}{N}\right|^2\nonumber\\
&=|f_A(0)|^2+\frac{|f'_A(0)|^2}{d}\sum_j\frac{E_j^2}{N^2}=|f_A(0)|^2+\frac{|f'_A(0)|^2\tr(H^2)}{N^2d}\overset{O(1/N^2)}=\frac{|\tr(HA)|^2}{d\tr(H^2)},
\end{align}
where we used (\ref{tail}), (\ref{asmpeq}), and Lemma \ref{l:2} in the steps marked with ``$O(q)$,'' ``$1/\poly(N)$,'' and ``$O(1/N^2)$,'' respectively. In the ``$\approx$'' step, we used (\ref{taylor}) with the approximation error upper bounded by
\begin{multline} \label{eq:me1}
O(1/d)\left|\sum_{j:|E_j|<\Lambda}\frac{E_j^3}{N^3}\right|+O(1/d)\sum_{j:|E_j|<\Lambda}\left(\frac{|f_A(0)|E_j^2}{N^2}+\frac{E_j^4}{N^4}\right)\\
\le O(1/d)\left|\sum_j\frac{E_j^3}{N^3}\right|+O(q)+O(1/d)\sum_j\left(\frac{E_j^2}{N^3}+\frac{E_j^4}{N^4}\right)=O(1/N^2),
\end{multline}
where we used Lemma \ref{l:m}. We complete the proof of Lemma \ref{l:thm} by combining (\ref{eq:mc}), (\ref{eq:me1}).

\printbibliography

\end{document}